\pdfoutput=1
\documentclass[pdftex,pra,aps,twocolumn,twoside,nofootinbib,showkeys,showpacs]{revtex4}

\usepackage{float}
\usepackage{indentfirst}
\usepackage{url}
\usepackage[stable]{footmisc}
\usepackage{fancyhdr}
\usepackage{appendix}
\usepackage[all]{xy}
\usepackage[pdftex]{graphicx}
\usepackage{graphicx}
\usepackage{youngtab}
\usepackage[T1]{fontenc}
\usepackage{textcomp}
\usepackage{newcent}
\usepackage[sc,noBBpl]{mathpazo}
\usepackage{amsmath,amsfonts,amssymb,amsthm}
\usepackage[mathscr]{eucal}
\theoremstyle{plain}
\newtheorem{theorem}{Theorem}
\newtheorem{lemma}[theorem]{Lemma}

\newtheorem{definition}[theorem]{Definition}
\newtheorem{fact}[theorem]{Fact}
\newtheoremstyle{note}{\topsep}{\topsep}{\slshape}{}{\scshape}{}{ }{}
\theoremstyle{note}

\newtheorem{remark}[theorem]{Remark}
\newcommand{\mbV}{\mathbb{V}}

\newcommand{\mbI}{\mathbb{I}}

\newcommand\cH{{\mathcal H}}

\newcommand\field{\mathbb}

\newcommand\C{\field{C}}

\renewcommand\Re{\operatorname{Re}}
\renewcommand\Im{\operatorname{Im}}

\newcommand\tr{\operatorname{Tr}}

\newcommand\rmi{\mathrm{i}\mspace{1mu}}

\newcommand{\<}{\langle}
\renewcommand{\>}{\rangle}

\newcommand\be{\begin{equation}}
\newcommand\ee{\end{equation}}
\newcommand\bea{\begin{array}}
\newcommand\eea{\end{array}}
\newcommand\ben{\begin{eqnarray}}
\newcommand\een{\end{eqnarray}}
\newcommand\ot{\otimes}

\newcommand\bei{\begin{itemize}}
\newcommand\eei{\end{itemize}}
\newcommand\bee{\begin{enumerate}}
\newcommand\eee{\end{enumerate}}

\begin{document}
\title{Region of fidelities for a $1\rightarrow N$ universal qubit quantum cloner}

\author{Piotr \'Cwikli\'nski$^{1,3,5}$,
Micha{\l} Horodecki$^{2}$ and Micha{\l} Studzi\'nski$^{2,4,*}$}
\affiliation{
$^1$ Faculty of Applied Physics and Mathematics, Gda\'nsk University of Technology, 80-233 Gda\'nsk, Poland \\
$^2$ Institute of Theoretical Physics and Astrophysics, University of Gda\'nsk, 80-952 Gda\'nsk, Poland \\
$^3$ School of Science and Technology, Physics Division, University of Camerino, I-62032 Camerino, Italy \\
$^4$ ICFO - Institut de Ciencies Fotoniques, Mediterranean Technology Park, 08860
Castelldefels (Barcelona), Spain \\
$^5$ Institute for Quantum Information, RWTH Aachen University, D-52056 Aachen, Germany
%$^4$ National Quantum Information Centre of Gda\'nsk, 81-824 Sopot, Poland%
}

\date{\today}
%\date{1 July 2009}

\begin{abstract}
We analyze a region of fidelities for qubit which is obtained after an application of a $1 \rightarrow  N$ universal quantum cloner. We express the allowed region for fidelities in terms of overlaps of pure states with irreps of $S_{n}$ ($n = N+1$) showing that the pure states can be taken with real coefficients only. Subsequently, the case $n = 4$, corresponding to a $1 \rightarrow 3$ cloner is studied in more details as an illustrative example. To obtain the main result, we make a convex hull of possible ranges of fidelities related to a given irrep. The formalism allows to construct the state giving rise to a given $N$-tuple of fidelities.
\end{abstract}


\maketitle
\let\oldthefootnote\thefootnote
\renewcommand{\thefootnote}{\fnsymbol{footnote}}
\footnotetext[1]{email: \url{studzinski.m.g@gmail.com}, tel: (+48 58) 551-20-34}
\let\thefootnote\oldthefootnote

\section{Introduction}
\label{sec:introduction}
\subsection{Historical overview: the no-cloning theorem}
In $1982$ Nick Herbert published a paper about faster than light communication, based on quantum correlations. He called his project $FLASH$, an acronym for 'First Laser-Amplified Superluminal Hookup' \cite{Herbert}. Of course, the idea was incorrect but it was the source for Wootters and \.Zurek \cite{WoottersZurek} and independently Dieks \cite{Dieks} to establish the so-called no-cloning theorem. The proof (we follow the version of the proof as in \cite{DAriano}, \cite{DAriano1}, \cite{DAriano2}) is the following; consider a cloning unitary operation $U_{cl}$ that is able to clone every given qubit. Before the copying process, we have states:
	\ben &&|A\> =|x\>|0\>|M\>, \nonumber \\
	     &&|B\> =|x^{'}\>|0\>|M\>, \label{MU} \een
where $|x\>$, $|x^{'}\>$ denote 'text' to copy for each state, $|0\>$ represents a 'blank card' and $|M\>$ is a cloning machine state ($M$ denotes the cloning machine). Acting $U_{cl}$ on (\ref{MU}), we should get:
\ben &&U_{cl}|A\> =|x\>|x\>|M^{'}\>, \nonumber \\
	     &&U_{cl}|B\> =|x^{'}\>|x^{'}\>|M^{''}\>, \label{MUU} \een
where $|M^{'}\>$, $|M^{''}\>$ represent new states of the machine.

Now let us consider a scalar product in both cases:
\be \< A|B\> = \< x|x^{'} \>, \label{modul1} \ee
since $\<0|0\> = \< M|M \> = 1$. On the other hand, for the second case, we have:
\be \< U_{cl}A|U_{cl}B\>=(\<x|x^{'}\>)^{2} \< M^{'}|M^{''}\>. \label{modul2} \ee

Comparing \eqref{modul1} and \eqref{modul2}, we get that in the case of $0 < | \<x|x^{'}\> | < 1$ scalar products do not match each other (but from the property of the unitary action, they should be the same). So, we have contradiction and it proves that one is not able to clone initially unknown quantum states. In general, the no-cloning theorem comes from the linearity of quantum mechanics \cite{WoottersZurek}, \cite{Dieks}.
Note that the no-cloning theorem could be generalized \cite{Gisin} or even strengthened \cite{Jozsa1}.
\renewcommand{\thefootnote}{\arabic{footnote}}

\subsection{Beyond the no-cloning theorem: origins of quantum cloning}	
Although, the no-cloning theorem is fundamental for quantum physics, it is not of great use in practice. The no-cloning theorem states that one is not able to copy an arbitrary quantum state. However, it is well known that performing ideal operations in physics and especially in quantum physics is impossible.
On the other hand, it is obvious that one can copy, though perhaps with a very bad quality.
Thus it is crucial to know the ultimate bounds for the quality of copying.

Several years after Wootters and \.Zurek paper has been published, namely in 1996, Hillery and Bu\u{z}ek published a paper called 'Quantum copying: beyond the no-cloning theorem' \cite{BuzekHillery}. It was the first time, when the above question regarding imperfect cloning has been formulated. Subsequently, the subject was a matter of wide research (see \cite{Gisin} and references therein for a comprehensive review). Here, we would like to point only a few (regarding the cloning of finite quantum states only). Quite soon, the Bu\u{z}ek-Hillery $1 \rightarrow 2$ (qubits) Quantum Cloning Machine ($QCM$) (for all formal definitions of quantum cloning machines, we refer to \cite{Gisin}), was generalized to the case $N_1 \rightarrow N_2$; first for qubits by Bru\ss{} et al. in \cite{Bruss-cloning1998}, and by Gisin and Massar in \cite{Massar}, and then for arbitrary-dimensional states by Werner in \cite{Werner-cloning1998}, and Keyl and Werner in \cite{KW}. Thus the family of symmetric Universal Quantum Cloning Machines ($UQCM$) is, at present, well known. What is more, the asymmetric $UQCM$ were also very heavily studied, let us mention here works of Braunstein et al. \cite{Braunstein-cloning2001}, Cerf \cite{Cerf-cloning2000}, Fiur{\'a}{\v s}ek et al. \cite{Fiurasek-cloning2005}, and Iblisdir et al. \cite{Iblisdir-cloning2004}, \cite{Iblisdir-cloning2005}. To emphasize the relevance of all these works, let us mention that formalism of asymmetric cloning machines is important in the context of quantum cryptography - it can be used in studies of relations between the eavesdropper's information gain and the noise in the channel. Later, efforts were made to unify these two kinds of $QCM$ \cite{Wang-cloning2011}. Of course, a lot of questions and problems still need answers, for instance, optimal state-dependent $QCM$ (as an example, see \cite{Chefles-cloning1999} or \cite{Siomau-cloning2010} - unfortunately, not many results are known for this kind of $QCM$) or optimal asymmetric $QCM$. There is also another 'gap' in quantum cloning: interestingly, up to our best knowledge, one is lacking a general result on an admissible region of fidelities (in general not the optimal one) for universal asymmetric $1 \rightarrow N$ quantum cloning machines; in our work, we want to make progress in this direction, providing an answer to this particular problem in the case of qubits. Let us mention briefly that it has been studied partially, for example in \cite{Fiurasek-cloning2005} (and partially in \cite{Iblisdir-cloning2004}), where expressions for an optimal fidelities region for an asymmetric $1 \rightarrow 3$ quantum cloner were obtained (and some partial results for a $1 \rightarrow N$ cloner).

In this Letter, we shall consider a $1 \rightarrow N$ universal quantum cloning machine (quantum cloner) for qubits. We first point out that this problem could be related to singlet states by recalling a relation between fidelities and singlet states. Using this fact, it turns out that by the application of Schur - Weyl duality, our problem could be quite easily solved and it leads to plots of ranges of fidelities for different irreducible representations of the symmetric group $S_{n}$ (where $n = N + 1$). The method of irreps (irreducible representations) is 'nice' here, because it makes calculations a lot easier. It allows us to decompose our initial Hilbert space into blocks of smaller dimensions, connected to a given partition $\lambda$ ($\lambda$'s label irreps of a symmetric group), and moreover, in our calculations, we can restrict ourself to a pure state only, linked to a given block. After taking the convex hull of the figures corresponding to all partitions, we can obtain the possible range of fidelities. In our case-study example of the $1 \rightarrow 3$ cloner, we check that our results are consistent with the existing methods for $1)$ symmetric cloners (for example, Keyl and Werner work \cite{KW}), namely we obtain that in the case of optimal, symmetric $UQCM$, result of all three fidelities equals $\frac{2}{3}$ is obtained, $2)$ an optimal asymmetric $1 \rightarrow 3$ cloner obtained in \cite{Fiurasek-cloning2005}. Finally, as a direct application of our method, we study one particular example, namely the case when one has $\mathop{\max}\limits_{F_1}\left(F_1+F_3=2F_2 \right)$.

This work is organized as follows. In Section \ref{sec:problem} we start with showing that the quantum cloning problem is equivalent to the entanglement sharing picture, so that the cloning machine is equivalently represented by a multipartite quantum state. Then, we formulate problem that we want to address in this Letter, namely, to determine the region of allowed fidelities for our $QCM$. What is more (and crucial for us), we show a transformation between Bell's states that allows to use Schur - Weyl duality. Section \ref{sec:young} is devoted to the group representation theory, in particular, Schur - Weyl decomposition. We present the basic formalism that leads to Schur - Weyl duality.
At the beginning, we introduce $SWAP$ operators and then we show how they lead to Schur - Weyl decomposition. At the end of this section, we show how Schur - Weyl duality is connected with Young diagrams formalism. This section is strictly mathematically oriented, but as it will be shown in the next section, this formalism allows us to simplify our problem a lot and it makes all calculations quite elementary. In Section \ref{sec:main}, which is the key section of our Letter, we present our main lemmas and theorems for $N-$tuples of fidelities, especially, we show which form of the multipartite state representing  $QCM$ can be used, and also how one can connect fidelity calculations with the formalism of Young diagrams. It follows, that to determine the admissible region of fidelities,
it is enough to  consider overlaps of real vectors with the matrices of irreps of the symmetric group. What is important, all results are valid for $N-$tuples of fidelities, so they describe an asymmetric $1 \rightarrow N$ cloner in general. In Section \ref{sec:case}, which is the case study of the $1 \rightarrow 3$ universal quantum cloner, we present how to obtain the allowed region for triples of fidelities. We also show that our results are consistent with: $a)$ calculations for a symmetric cloning case, predicted by the Werner's formula \cite{Werner-cloning1998} and $b)$ partially with results obtained in \cite{Fiurasek-cloning2005}, where optimal fidelities for this kind of $UQCM$ were presented. At the end, we give an application of our model, namely we show that for any given triple of fidelities from the allowed region, we could reconstruct a state that gives rise to that fidelities. This technique is actually quite general and could be applied to the $N$ number of clones. In Section \ref{sec:S5}, we briefly explore our approach with a higher number of clones, by studying the case of a $1 \rightarrow 4$ $UQCM$.
%in this sense that we could consider not only the symmetric cloning case, but also the asymmetric one. \tred{czemu z mozliwosci rekonstrukcji stanu, wynika
%ze nasza maszyna jest "uniwersalna"? }{\color{green} Troszke niefortunnego slowa uzylem (zmienilem na general) - chodzilo mi w tym zdaniu o to, ze nasza maszynka to nie jest taka zwykla quantum-shmantum rzecz, ktore jest tylko symetryczna badz antysymetryczna - jest uniwersalna, bo lacze te oba przypadki, nie chodzilo mi o to, ze jest uniwersalna jako z definicji uniwersalnej maszyny klonujacej:D}

\section{Statement of the problem}
\label{sec:problem}
 Let us recast a question of cloning in an equivalent picture of entanglement sharing (for simplicity, we focus here on our case study-example, namely the $1 \rightarrow 3$ $UQCM$ - all calculation could be easily extended to the case of $N$ clones). Suppose that we have an initial state described by the Bell state: $| \psi^{+} \> = \sqrt{\frac{1}{2}} \left( |00\> + |11\> \right)$. We also have a cloning machine $M$, which could be described by a completely positive, trace preserving (CPTP) map $\widetilde{\Lambda}$ (note that the cloning map $\widetilde{\Lambda}$
applied to the second subsystem of the maximally entangled state $|\psi^{+} \>$, when the first is untouched, produces, in general, mixed state that contains all the information about the map). As an output, we want to obtain $N$ shares of our initial state.
%each copy is strongly entangled with initial state (as it was shown in \cite{BuzekHillery}).%
Our scheme is presented in Figure \ref{fig:cloning} (for our case-study example).
\begin{figure}[h!]
\begin{center}
\includegraphics[width=2.5in]{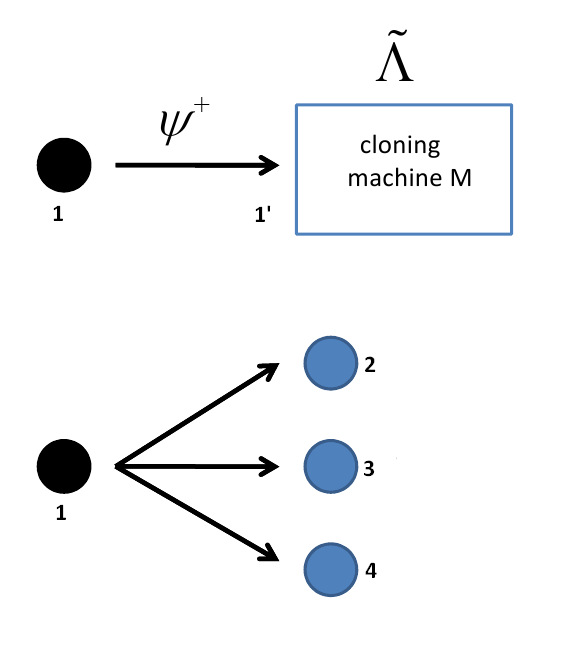}
\end{center}
\caption{Situation before and after cloning.}
\label{fig:cloning}
\end{figure}
%Let us write state connected to the cloning machine as
Let us write a state obtained after an application of the cloning map to the half of $|\psi\rangle^+$
\be \rho_{1234} = \frac{1}{2} \left( \operatorname{\mbI} \ot \widetilde{\Lambda} \right)\left( | \psi^{+} \> \< \psi^{+} | \right), \label{map} \ee
where indexes $1$, $2$, $3$ and $4$ are connected to an initial state and clones, according to Figure \ref{fig:cloning}.

We want to calculate an allowed region for singlet fractions $F_{1i}$ (between the initial state and one of the three clones (copies)), denoted by \cite{singlet}
\ben
&&F_{12} = \< \psi_{12}^{+} | \tr_{34} (\rho_{1234}) | \psi_{12}^{+} \>, \nonumber \\
&&F_{13} = \< \psi_{13}^{+} | \tr_{24} (\rho_{1234}) | \psi_{13}^{+} \>, \nonumber \\
&&F_{14} = \< \psi_{14}^{+} | \tr_{23} (\rho_{1234}) | \psi_{14}^{+} \>.
\een
Note that:
\be
F_{1i} = \< \psi^{-}_{1i} | \tr_{\overline{1i}}(\widetilde{\rho}_{1234}) | \psi^{-}_{1i} \>
\label{note},
\ee
where $\tr_{\overline{1i}}$ means partial trace over all systems except $1i$,
and $|\psi^{-}_{1i}\>$ and $\widetilde{\rho}_{1234}$ are defined below.
The vector $|\psi^{-}_{1i}\> = U \ot \operatorname{\mbI} | \psi^{+}_{1 \widetilde{1}} \>$, $| \psi^{-} \>$ is obtained after the action of $-i\sigma_{y}$ on $|\psi^{+}\>$:

\begin{equation}
\begin{split}
 \label{sigmay} | \psi^{-} \> &= -i \sigma_y | \psi^{+} \> = -i \sigma_y \sqrt{\frac{1}{2}} \left( |00\> + |11\> \right) =\\
&= \sqrt{\frac{1}{2}} \left( |01\> - |10\> \right),
\end{split}
\end{equation}
where $\sigma_y$ is one of the Pauli matrices: $\sigma_y = \begin{bmatrix}
0 & -i\\
i & 0
\end{bmatrix}$. %State $\sqrt{\frac{1}{2}} \left( |01\> - |10\> \right)$ is another example of Bell state.%
Using \eqref{sigmay} we can write that:
\be |\psi^{-}_{1 \widetilde{1}} \> = U \ot \operatorname{\mbI} |\psi^{+}_{1\widetilde{1}} \>, \ee
where $U = -i \sigma_y$.
The state $\widetilde{\rho}_{1234}$ from Eq. \eqref{note} is obtained after the following transformation:
\be
\begin{split}
\widetilde{\rho}_{1234} &= ( \operatorname{\mbI} \ot \widetilde{\Lambda} )|\psi^{-}_{1\widetilde{1}} \> \\ &=
( U \ot \operatorname{\mbI} )  \left(( \operatorname{\mbI} \ot \widetilde{\Lambda} )|\psi^{+}_{1\widetilde{1}}\>   \< \psi^{+}_{1\widetilde{1}} | \right) ( U^{\dag} \ot \operatorname{\mbI}).
\end{split}
\ee
The four-partites states $\widetilde{\rho}_{1234}$, with the constraint $\widetilde{\rho}_1=I/2$,
are in one-to-one correspondence with cloning machines, and the cloning fidelity of a given
machine is determined by fidelities $F_{1i}$ of the corresponding state.

Note that for any channel $\Lambda$, fidelity of the state $\left( \operatorname{\mbI} \ot \widetilde{\Lambda} \right)\left( | \psi^{+} \> \< \psi^{+} | \right)$ can be related to the average fidelity of transmission
of an initial state as follows~\cite{MHPH}:
\be f = \frac{Fd + 1}{d+1}, \label{mipa} \ee
where $d$ is the dimension of the Hilbert space $\cH \cong \C^{d}$. Thus instead of cloning fidelities $f$,
we can consider singlet fractions $F$, which we will further call simply fidelities,
while the fidelity of cloning we will term the 'cloning fidelity'.
%Now, we want to make a crucial remark concerning our work, because of Eq. \eqref{mipa}, starting from now
%on, we will be using term 'fidelity', in the sense of the 'singlet fraction' (unless it is stated
%differently), we hope that this will not lead to any confusions.

Werner \cite{Werner-cloning1998} provided the following formula for an optimal cloning fidelity
of universal symmetric $N_1 \rightarrow N_2$ cloning machine:
%we can relate singlet fraction with fidelity (in the case of the optimal fidelity for universal symmetric $N
 %- M$ cloning machine) as follows:
\be
f_{N_1 N_2}(d) = \frac{N_1}{N_2} + \frac{(N_2-N_1)(N_1+1)}{N_2(N_1+d)}.
\ee
In our case, when $N_1 = 1$, $N_2 = 3$ and $d = 2$, we obtain that the fidelity $f$ for the universal, symmetric cloning machine should be equal to
\be
f = \frac{1}{3} + \frac{4}{9} = \frac{7}{9}. \label{werner}
\ee

We can now formulate question for our case study-example: which values of triples of cloning fidelities
$\left( f_{12}, f_{13}, f_{14} \right)$ are allowed for a universal cloning machine? (Of course in general our question is the following: which values of $N-$tuple of cloning fidelities
$\left( f_{12}, f_{13}, \ldots, f_{1n} \right)$ are allowed for a universal cloning machine?)
As said above, we shall address equivalent question: what values of triples of fidelities
$\left( F_{12}, F_{13}, F_{14} \right)$ are allowed for an arbitrary state
of a maximally mixed first subsystem?
In the next sections, the answer is presented.

%\tred{W poprzednich sekcjach - $\pi$ nie powinno byc w nawiasach - bo sie myli z cyklem.}
%\section{The new model}
\section{Mathematical introduction: Schur - Weyl decomposition}
\label{sec:young}
In this section we introduce necessary mathematical tools from group theory. We are especially focused on Schur-Weyl duality \cite{Audenaert}.

Consider a unitary representation of a permutation group $S_n$ acting on the $n-$fold tensor product of complex spaces $\mathbb{C}^d$, so our full Hilbert space is $\mathcal{H}\cong (\mathbb{C}^d)^{\otimes n}$. For a fixed permutation $\pi\in S_n$ a unitary transformation $V_{\pi}$ is given by
\be
\label{unitary}
V_{\pi}\left( |i_1\>\otimes \ldots  \otimes |i_n\>\right)=|i_{\pi(1)}\> \otimes  \ldots  \otimes |i_{\pi(n)}\>,
\ee
where $|i_1\>,\ldots,|i_n\>$ is a standard basis in $(\mathbb{C}^d)^{\otimes n}$. The space of rank$-n$ tensors can be also consider as a representation space for a general linear group $\operatorname{GL}(d,\mathbb{C})$. Let $U\in \operatorname{GL}(d,\mathbb{C})$, thus, this  induces in the tensor product $(\mathbb{C}^d)^{\otimes n}$ the following transformation
\be
\label{unit}
U^{\ot n}\left(|i_1\>\otimes  \ldots  \otimes |i_n\>\right)=U|i_1\> \otimes  \ldots \otimes U|i_n\>.
\ee
A key property is that these two representations turn out to be each other commutants.  Any operator  on $(\mathbb{C}^d)^{\otimes n}$ that commutes with all $U^{\otimes n}, \ \forall U\in \operatorname{GL}(d,\mathbb{C})$, is a linear combination of permutation matrices $V_{\pi}$. Conversely, any operator commuting with all permutation matrices $V_{\pi}, \ \forall \pi\in S_n$, is a linear combination of $U^{\otimes n}$. This duality is called Schur-Weyl duality. It was shown (see, for example, \cite{Harrow-PhD} or \cite{Harrow-Schur2006}) that there always exists some basis called the Schur basis which gives decomposition of $V_{\pi}$ and  $U^{\otimes n}$ into irreducible representations (irreps) simultaneously. Thanks to this, the space $(\mathbb{C}^d)^{\otimes n}$ can be decomposed into irreducible representations of $S_n$
\be
\label{decomp}
(\mathbb{C}^d)^{\otimes n}\cong \bigoplus_{\lambda \vdash n} \mathcal{H}_{\lambda}^{\mathcal{U}}\otimes \mathcal{H}_{\lambda}^{\mathcal{S}},
\ee
where $\lambda$ labels inequivalent irreps of $S_n$ and $\mathcal{H}_{\lambda}^{\mathcal{U}}$ is the multiplicity space. It is called Schur-Weyl decomposition. The labels $\lambda$ are allowed partitions of some natural number $n$. Every partition is a sequence $\lambda=(\lambda_1,\ldots,\lambda_r)$ satisfying
\be
\label{partition}
\forall_{i} \  \lambda_i \geq 0,\quad \lambda_1\geq \lambda_2\geq \ldots \geq \lambda_r,\quad \sum_{i=i}^r\lambda_i=n,
\ee
where $r\in \{1,\ldots,n\}$. Every such partition corresponds to some diagram, which is called the Young diagram \cite{Fulton1991-book-rep}. Here are few examples of Young diagrams for $n=4,6$ and $3$ respectively:
\[
\begin{split}
&\yng(2,2)\qquad \qquad \yng(3,2,1)\qquad  \qquad \qquad \yng(1,1,1)\\
 \ \lambda&=(2,2),\qquad \lambda=(3,2,1),\qquad \  \lambda=(1,1,1)\\
%\text{A few}&\text{ examples of Young diagrams with corresponding partitions $\lambda$. }
\end{split}
\]
In this Letter we are interested in representations on the symmetric part $\mathcal{H}^{\mathcal{S}}_{\lambda}$. For example the $SWAP$ operator $\mathbb{V}_{\pi}$ can be decomposed, due to Schur - Weyl decomposition, in the following way:
\begin{equation}
\label{eq:swapdec}
\mathbb{V}_{\pi}=\bigoplus_{\lambda} \operatorname{\mbI}_{r(\lambda)} \ot \widetilde{\mbV}_{\pi}^{\lambda},
\end{equation}
where $\pi \in S_n$ and $r(\lambda)$ is the dimension of a unitary part.
The operators $\widetilde{\mbV}_{\pi}^{\lambda}$ are irreducible representations
of $S_n$.
 Thanks to the above-mentioned method, we can decompose $U^{\otimes n}$-invariant states in the following way:
\begin{equation}
\label{general2}
\rho_{1\ldots n}=\bigoplus_{\lambda} \mbI_{r(\lambda)}\ot \widetilde{\rho}^{\lambda}.
\end{equation}
Note that fidelities with singlet states are invariant under averaging over $U \ot U$ transformations.
Therefore the $N-tuple$ of fidelities is invariant under an application of transformation $U^{\ot n}$
to the state, so we can always use density operators that are commutant of $U^{\ot n}$ i.e.
they are of the form~\eqref{general2}.

\section{Expression for an allowed region of $N$-tuples of fidelities.}
\label{sec:main}
In this section we provide a general formula for an allowed region of $N$-tuples
of fidelities in terms of overlaps of pure states with irreducible representations of $S_n$ ($n = N+1$).
This is contained in Theorem \ref{thm:main}.
We further show in Lemma \ref{real} that one can restrict attention to pure states
with real coefficients, since they determine the admissible region of fidelities.

%\subsection{Useful definitions and lemmas}
% \begin{remark}
% Every $n-$particle state $\rho_{1\ldots n}$ can be decomposed in following way
% \begin{equation}
% \label{general}
% \rho_{1\ldots n}=\bigoplus_{\lambda} \rho^{\mathcal{U}}_{\lambda}\ot \rho^{\mathcal{S}}_{\lambda},
% \end{equation}
% where $\rho^{\mathcal{U}}_{\lambda}$ is density operator on unitary part and $\rho^{\mathcal{S}}_{\lambda}$ is density operator on symmetry part.
% \end{remark}
\begin{lemma}
\label{FF}
%Fidelity $F_{1k}$ between an initial state of our $QCM$ which is of the form~\eqref{general2} and its 'k-th' %clone of the initial state could be written as
Fidelity $F_{1k}$ as defined in \eqref{note} is of the form
\be
\label{FFF}
F_{1k}=\sum_{\lambda}F_{1k}^{\lambda},
\ee
where
\be
\label{ff}
F_{1k}^{\lambda}
%=\frac{1}{2}-\frac{1}{2}\tr \left(\rho^{\lambda}\widetilde{\mbV}_{(1k)}^{\lambda} \right)
= \frac{1}{2}- \frac{1}{2} \tr \left( \rho^{\lambda} \widetilde{\mbV}_{(1k)}^{\lambda} \right),
\ee
The lower index $(1k)$ means a permutation that swaps $1$ and $k$,
and $\rho^\lambda$'s are arbitrary normalized states on partition $\lambda$.
\end{lemma}

\begin{proof}
From the definition of a fidelity we can write
\be
F_{1k}=\langle \psi_{1k}|\rho_{1k}|\psi_{1k}\rangle=\tr\left(\rho_{1k}|\psi_{1k}\rangle\langle \psi_{1k}| \right), \label{nF}
\ee
where $|\psi_{1k}\> \< \psi_{1k}| = \frac{1}{2}(id_{1k} - \mbV_{1k})$, $\rho_{1k} = \tr_{\overline{1k}} \rho_{1\ldots n}$ and $\tr_{\overline{1k}}$ denote partial trace over all
systems except $1$ and $k$. Expanding \eqref{nF}, we obtain:
\ben \label{nF1}
F_{1k} &=& \tr \left( \frac{1}{2} \rho_{1k} (id_{1k} - \mbV_{(1k)}) \right) = \tr \left( \frac{1}{2} \rho_{1k} - \frac{1}{2}\rho_{1k} \mbV_{(1k)} \right) \nonumber \\
&=& \frac{1}{2} - \frac{1}{2}\tr \left( \mbV_{(1k)} \rho_{1\ldots n} \right). \een
Now we can use Schur - Weyl decomposition to represent $\mbV_{(1k)}$ and $\rho_{1\ldots n}$:
\be \mbV_{(1k)}=\bigoplus_{\lambda} \operatorname{\mbI}_{r(\lambda)} \ot \widetilde{\mbV}_{(1k)}^{\lambda}, \ \ \rho_{1\ldots n}=\bigoplus_{\lambda} \operatorname{\mbI}_{r(\lambda)} \ot \widetilde{\rho}^{\lambda} \label{Vr}, \ee
Inserting \eqref{Vr} into \eqref{nF1}, we have:
\begin{equation}
\label{oF}
\begin{split}
F_{1k} &= \frac{1}{2} - \frac{1}{2} \tr \left(\left(\bigoplus_{\lambda} \operatorname{\mbI}_{r(\lambda)} \ot \widetilde{\rho}^{\lambda}\right)\left(\bigoplus_{\mu} \operatorname{\mbI}_{r(\mu)} \ot \widetilde{\mbV}_{(1k)}^{\lambda}\right)\right)= \\
&= \sum_{\lambda}\left( \frac{1}{2}- \frac{1}{2} \tr (\rho^{\lambda} \widetilde {\mbV}^{\lambda}_{(1k)}) \right).
\end{split}
\end{equation}
%for a fixed $\lambda$.
Equation \eqref{oF} could be rewritten as:
\be \label{part} F_{1k} = \sum_{\lambda} F_{1k}^{\lambda}, \ee
where $F_{1k}^{\lambda} = \frac{1}{2} - \frac{1}{2}\tr \left(\rho^{\lambda}\widetilde{\mbV}^{\lambda}_{(1k)} \right)$,
and $\rho^{\lambda} = d_{\lambda} \widetilde{\rho}^{\lambda}$ and $d_{\lambda}$ stands for the dimension of a given partition.
%One can also see that $\sum_{\lambda} \tr (\rho^{\lambda})=1$, which is in fact our normalization relation. Note finally, that arbitrary states $\rho^\lambda$ satisfying
%the above normalization give rise to some state $\rho_{1\ldots n}$. Moreover,
%any state of the form \eqref{general2} has the subsystem $1$. Therefore,
%an arbitrary $k$-tuple of fidelities $F_{1k}$ corresponds to some set of  $\rho^\lambda$'s
%satisfying the above normalization.

%One can see that $G^{\lambda}_{1k}$ consists knowledge about Young tableaus. To obtain full knowledge about our model, also calculations according to \eqref{nF} and \eqref{part} should be done for $F_{12}$ and $F_{14}$, so we should calculate:
%\be F_{1i} = \sum_{\lambda} F_{1i}^{\lambda} \label{fpart} \ee
\end{proof}

 Now we are in position to formulate the main theorem of this section:
\begin{theorem}
\label{thm:main}
The set $\mathcal{F}$ of admissible vectors of fidelities $\left\{ F_{12}, \ldots, F_{1n} \right\}$ is of the form
\be
\mathcal{F} = \operatorname{conv} \left( \bigcup_\lambda  \mathcal{F}^\lambda \right),
 \ee
where $\operatorname{conv}$ stands for a convex hull, the union runs over all irreps of $S_n$
 and
\be  \mathcal{F}^\lambda = \left\{ \left( F_{12}^{\lambda},\ldots, F_{1n}^{\lambda}\right) \ : |\psi\rangle \in \C^{d_{\lambda}}\  \right\}, \ee \label{maintheo}
where $F_{1k}^{\lambda}$ are of the form: $F_{1k}^{\lambda} = \frac{1}{2} - \frac{1}{2}\< \psi |\widetilde{{\mbV}}^{\lambda}_{(1k)} |\psi \>$, and where  $|\psi \>$ is a pure state.
\end{theorem}

\begin{proof}
At the beginning, let us consider the following mapping:
\be
\vec{F}: \ \operatorname{P} \rightarrow \mathcal{R}^{N}
\ee
which maps states $\rho_{1\ldots n} \in \operatorname{P}$ ($\operatorname{P}$ stands for a convex set of all states $\rho_{1\ldots n}$) into the $N$-tuples $(F_{12}, \ldots, F_{1n}) \in \mathcal{R}^N$. Explicitly, we have
\be \label{RHS}
\vec{F}(\rho_{1\ldots n}) = \left[ F_{12}(\rho_{1\ldots n}), \ldots, F_{1n}(\rho_{1\ldots n}) \right].
\ee
This mapping is affine (Lemma~\ref{affinite}), i.e. is of the form:
\be \label{ee}
\vec{F}(\rho_{1\ldots n}) = \widetilde{\vec{F}}(\rho_{1\ldots n}) + \vec{C}, \ee
where $\widetilde{\vec{F}}: \operatorname{P} \rightarrow \mathcal{R}^N $ is linear.
Indeed, one can see that $RHS$ of \eqref{RHS} could be written as:

\be
\label{een}
\begin{split}
F_{12}(\rho_{1\ldots n})&=\frac{1}{2}\tr\left[ (\mbV_{(1)\ldots (n)}-\mbV_{(12)})\rho_{1\ldots n}\right]=\\
&=\frac{1}{2}-\frac{1}{2}\tr(\mbV_{(12)}\rho_{1\ldots n}),\\
& \qquad \qquad \vdots \\
F_{1n}(\rho_{1\ldots n})&= \frac{1}{2}\tr\left[(\mbV_{(1)\ldots (n)}-\mbV_{(1n)})\rho_{1\ldots n}\right]=\\
&=\frac{1}{2}-\frac{1}{2}\tr(\mbV_{(1n)}\rho_{1\ldots n}).
\end{split}
\ee

Comparing \eqref{ee} and \eqref{een}, we obtain that in our case $\vec{C} = \left[ \frac{1}{2}, \ldots, \frac{1}{2} \right]$ and $\widetilde{\vec{F}}(\rho_{1\ldots n}) = -\left[\tr (\mbV_{(12)}\rho_{1\ldots n}), \ldots, \tr (\mbV_{(1n)}\rho_{1\ldots n}) \right]$, which is obviously linear with respect to $\rho_{1\ldots n}$.

One can note that in general $\widetilde{\rho}^{\lambda}$ from Eq. \eqref{Vr} is a mixed state, but according to Lemma \ref{affinite}, we can take the mapping of extreme points of $\rho_{1\ldots n}$ of the form:
\be
\label{horo}
\begin{split}
\rho_{1\ldots n} &= \bigoplus_{\lambda} \operatorname{\mbI}_{r(\lambda)} \ot \widetilde{\rho}^{\lambda} = \bigoplus_{\lambda}  \frac{1}{d_{\lambda}} \operatorname{\mbI}_{r(\lambda)} \ot d_{\lambda} \widetilde{\rho}^{\lambda}  \\
&= \bigoplus_{\lambda} \frac{1}{d_{\lambda}} \operatorname{\mbI}_{r(\lambda)} \ot \rho^{\lambda} = \bigoplus_{\lambda} p_{\lambda} \left( \frac{\operatorname{\mbI}_{r(\lambda)}}{d_{\lambda}}  \ot \hat{\rho}^{\lambda} \right),
\end{split}
\ee
where $\hat{\rho}^{\lambda} = \frac{\rho^{\lambda}}{\tr \rho^{\lambda}}$, so namely, it is normalized, and $p_{\lambda} = \tr \rho^{\lambda}$. We can now eigen-decompose $\hat{\rho}^{\lambda}$:
\be \hat{\rho}^{\lambda} = \sum_{i=1}^{n_{\lambda}} p_{i}^{\lambda} | \psi^{\lambda}_{i} \> \< \psi^{\lambda}_{i} |, \label{dec} \ee
where $\sum_{\lambda} p_{\lambda} = 1$. \\
Inserting Eq. \eqref{dec} into Eq. \eqref{horo}, we get that:
\be \label{ex}
\bigoplus_{\lambda} \sum_{i=1}^{n_{\lambda}} p_{\lambda} p_{i}^{\lambda} \left( \frac{\operatorname{\mbI}_{r(\lambda)}}{d_{\lambda}} \ot | \psi_{i}^{\lambda} \> \< \psi_{i}^{\lambda} | \right), \ee
we see from \eqref{ex} that extreme points are of the form:
\be \label{lipka}
\rho_{extreme}^{\lambda} = \frac{\operatorname{\mbI}_{r(\lambda)}}{d_{\lambda}} \ot | \psi^{\lambda} \> \< \psi^{\lambda} |, \ee
where $\lambda$ runs over all irreps from $S_n$ and $| \psi^{\lambda} \>$ is an arbitrary state form the space of irrep linked to a given partition $\lambda$. Inserting Eq. \eqref{lipka} into Eq. \eqref{oF}, we obtain the desired result.
%where as a state $\psi^{\lambda}$ we insert $\psi^{(2,2)} \in \mathcal{C}^2$ and $\psi^{(3,1)} \in \mathcal{C}^3$, respectively.
%\end{proof}
%From Eq. \eqref{lipka}, we are able to conclude that possible are values of $n-$tuples $(F_{12}, \ldots, F_{1n})$ of the %fidelities with initial state of the form $\rho_{1 \ldots n} = \left( \operatorname{\mbI} \ot \widetilde{\Lambda} \right) %|\psi^{+}\> \< \psi^{+} |$).
% that form a figure which is bounded from inside by the three-dimensional figure from Figure (\ref{fig:ch}).
\end{proof}

We note that to determine the allowed region of fidelities, it is enough to consider only
vectors of real coefficients.
\begin{lemma}
\label{real}
To generate a convex hull of the allowed region of fidelities, it is sufficient to consider pure states of real coefficients only.
\end{lemma}
\begin{proof}
We need to show that, in our case, a kind of majorization of complex pure states by real ones occurs. To prove that, note that our operators $\widetilde{\mbV}^{\lambda}_{(1i)}$ (\ref{opV}) are symmetric and real\footnote[1]{Thanks to \cite{Chen2002}, we know that swap operators $\widetilde{\mbV}^{\lambda}_{(i,i+1)}$ are symmetric and real. It easy to see that every swap operator between system 1 and $i$ is of the following form: $\widetilde{\mbV}^{\lambda}_{(1i)}=\left(\widetilde{\mbV}^{\lambda}_{(12)} \widetilde{\mbV}^{\lambda}_{(23)}\cdot \ldots \cdot \widetilde{\mbV}^{\lambda}_{(i-2,i-1)}\right)\widetilde{\mbV}^{\lambda}_{(i-1,i)}\left(\widetilde{\mbV}^{\lambda}_{(i-2,i-1)}\cdot \ldots \cdot \widetilde{\mbV}^{\lambda}_{(23)}\widetilde{\mbV}^{\lambda}_{(12)}\right)=A \widetilde{\mbV}^{\lambda}_{(i-1,i)}A^{T} $. We know that for every symmetric matrix $S$ and an arbitrary nonzero matrix $X$, their composition $XSX^{T}$ is also symmetric. Thanks to this, the swap operator $\widetilde{\mbV}^{\lambda}_{(1i)}$ is an symmetric matrix.}, so they could be written in a general form as:
\begin{equation}
\label{123}
\widetilde{\mbV}^{\lambda}_{(1i)}=\begin{bmatrix}
v_{11} & v_{12} & \ldots & v_{1k}\\
v_{12} & v_{22} & \ldots & v_{2k}\\
\vdots & \vdots & \ddots &\vdots \\
v_{1k} & v_{2k} & \ldots & v_{kk}
\end{bmatrix},
\end{equation}
where $k=\operatorname{dim}\mathcal{H}^{\mathcal{S}}_{\lambda}$. Now let us write a density matrix of a pure state $|\psi \rangle=(a_1,a_2,\ldots , a_k)^{T}$ with complex values (letter $C$ corresponds to the word 'complex'):
\begin{equation}
\label{12}
\rho_C^{\lambda}=\begin{pmatrix}
|a_1|^2 & a_1a_2 & \ldots & a_1a_k \\
\bar{a}_1\bar{a}_2 & |a_2|^2 & \ldots & a_2a_k \\
\vdots & \vdots & \ddots & \vdots \\
\bar{a}_1\bar{a}_k & \bar{a}_2\bar{a}_k & \ldots & |a_k|^2
\end{pmatrix}.
\end{equation}
Next, let us rewrite $\tr \left( \rho^{\lambda} \widetilde{\mbV}^{\lambda}_{(1i)} \right)$ from Lemma \ref{FF}, using (\ref{12}) and (\ref{123}), as follows:
\begin{equation}
\label{michal}
\begin{split}
\tr \left(\rho_C^{\lambda}\widetilde{\mbV}^{\lambda}_{(1i)} \right)&=\sum_{j=1}^k |a_j|^2v_{jj}+\sum_{j,s=1}^k\left(a_ja_s+\bar{a}_j\bar{a}_s \right)v_{js}=\\
&=\sum_{j=1}^k |a_j|^2v_{jj}+2\sum_{j,s=1}^k \operatorname{Re}\left(a_ja_s \right)v_{js}.
\end{split}
\end{equation}
We see that calculating the trace~\eqref{michal} with the operator $\rho_C^{\lambda}$ is equivalent to calculating the trace~\eqref{michal} with the following matrix $\rho_{R-C}^{\lambda}$:
\be
\label{piotr}
\rho^{\lambda}_{R-C} = \begin{bmatrix}
|a_1|^2 & \Re(a_1a_2) & \ldots & \Re(a_1a_k)\\
\Re(a_1a_2) & |a_2|^2 & \ldots & \Re(a_2a_k)\\
\vdots & \vdots & \ddots & \vdots \\
\Re(a_1a_k) & \Re(a_2a_k) & \ldots & |a_k|^2
\end{bmatrix},
\ee
where the new index $R-C$ means that we are left only with real numbers. We shall now show, that $\rho^{\lambda}_{R-C}$ is positive-semidefinite.

Let us denote the total phase in this case by $e^{\rmi \phi_j}$ for $j=1,\ldots,k$, so we are left with $\Re(e^{\rmi \phi_j}) = \Re(\cos \phi_j + \rmi \sin \phi_j) = \cos \phi_j$. We then obtain
\be
\label{piotrn}
\begin{split}
\rho^{\lambda}_{R-C} &= \begin{bmatrix}
|a_1|^2 & |a_1a_2| & \ldots & |a_1a_k|\\
|a_1a_2| & |a_2|^2 & \ldots & |a_2a_k|\\
\vdots & \vdots & \ddots & \vdots \\
|a_1a_k| & |a_2a_k| & \ldots & |a_k|^2
\end{bmatrix}\bullet \\
&\begin{bmatrix}
1 & \cos(\phi_1-\phi_2) & \ldots & \cos(\phi_1-\phi_k)\\
\cos(\phi_1-\phi_2) &1 & \ldots & \cos(\phi_2-\phi_k)\\
\vdots & \vdots & \ddots & \vdots \\
\cos(\phi_1-\phi_k) & \cos(\phi_2-\phi_k) & \ldots & 1
\end{bmatrix}=\\
&=A\bullet C,
\end{split}
\ee
where $\phi_j$ are phases connected with a given $|a_j|$ and by $\bullet$ we denote a Hadamard product of two matrices.
Now our goal is to show that the matrices $A$ and  $C$ are positive-semidefinite.
First let us define a set of vectors:
\be
|\omega_j \rangle=\cos \phi_j |0 \rangle+\sin \phi_j |1 \rangle, \quad \text{for} \quad j=1,2,\ldots,k.
\ee
It is easy to show that  matrix $C$ can be rewritten in terms of vectors $|\omega_l\rangle$ i.e. $C_{ij}= \langle \omega_i|\omega_j \rangle $ for $i,j=1,2,\ldots,k$. Thanks to this operation we can conclude that $C$ is a square Gramian matrix. From~\cite{Horn1985-book-mat} we know that every square Gramian matrix is positive-semidefinite.

The matrix $A$ is of the form $|\phi\>\<\phi|$, with $|\phi\rangle=\sum_i |a_i| |i\>$, so it is positive-semidefinite. Now using Fact~\ref{Had} we obtain that our operator $\rho_{R-C}^{\lambda}$ is also positive-semidefinite.

Because $\rho_{R-C}^{\lambda}$ is real and symmetric we get from Fact~\ref{RS} that the matrix $\rho^{\lambda}_{R-C}$ possesses real eigenvectors, so indeed it is a mixture of real pure states.

%Since $C$ and $A$ are positive-semidefinite matrices then thanks to fact~\ref{Had} (see Appendix) an operator $\rho^{\lambda}_{R-C}$ from Eq. \eqref{piotrn}
\end{proof}

\section{Case study: a region for triples of fidelities}
\label{sec:case}
\subsection{Partitions and transpositions}
In our case-study example we have four particles %(or equivalently: we have four-partite state), so the corresponding dimension is
($n=4$) which means that allowed partitions are $\lambda_1=(4), \lambda_2=(3,1), \lambda_3=(2,2), \lambda_4=(2,1,1)$ and $\lambda_5=(1,1,1,1)$. Because we want to consider only qubits here, so in our case $d=2$, and hence $\lambda$ runs over binary partitions only or, equivalently,
over Young diagrams with two rows (for other diagrams the multiplicity space $\mathcal{H}^{U}_{\lambda}$ becomes zero-dimensional). Thanks to this, we are left only with $\lambda_1=(4), \lambda_2=(3,1), \lambda_3=(2,2)$.
For the partition $\lambda_2=(3,1)$  unitary representations of transpositions $\operatorname{T}(1,2)$, $\operatorname{T}(1,3)$, $\operatorname{T}(1,4)$, which we shall need, are \cite{Thomas1980}:
\begin{equation}
\begin{split}
\label{opV}
\widetilde{\mbV}_{(12)}^{\lambda_2}&=
\begin{bmatrix}
1 & 0& 0 \\ 0&1&0\\ 0&0&-1
\end{bmatrix}, \ \
\widetilde{\mbV}_{(13)}^{\lambda_2}=\begin{bmatrix}
1&0&0\\ 0&-\frac{1}{2}&-\frac{\sqrt{3}}{2}\\ 0&-\frac{\sqrt{3}}{2}&\frac{1}{2}
\end{bmatrix}, \\
\widetilde{\mbV}_{(14)}^{\lambda_2}&=\begin{bmatrix}
-\frac{1}{3}&-\frac{\sqrt{2}}{3}&-\frac{\sqrt{6}}{3}\\-\frac{\sqrt{2}}{3}&\frac{5}{6}&-\frac{\sqrt{3}}{6}\\ -\frac{\sqrt{6}}{3}&-\frac{\sqrt{3}}{6}&\frac{1}{2}
\end{bmatrix}.
\end{split}
\end{equation}
For the partition $\lambda_3=(2,2)$ unitary transposition representations $\operatorname{T}(1,2)$, $\operatorname{T}(1,3)$, $\operatorname{T}(1,4)$, which we shall need, are \cite{Thomas1980}:
\begin{equation}
\begin{split}
\widetilde{\mbV}_{(12)}^{\lambda_3}&=
\begin{bmatrix}
1 & 0\\
0 & -1
\end{bmatrix},\ \
\widetilde{\mbV}_{(13)}^{\lambda_3}=\begin{bmatrix}
-\frac{1}{2} & -\frac{\sqrt{3}}{2}\\
-\frac{\sqrt{3}}{2} & \frac{1}{2}
\end{bmatrix},\\
\widetilde{\mbV}_{(14)}^{\lambda_3}&=
\begin{bmatrix}
-\frac{1}{2} & \frac{\sqrt{3}}{2}\\
\frac{\sqrt{3}}{2} &  \frac{1}{2}
\end{bmatrix},
\end{split}
\end{equation}
%where $\widetilde{\mbV}^{\lambda}_{(1i)}$ corresponds to a given pair initial state - clone (for example $13$).
Note that for the partition $\lambda_1$ we have a trivial representation, so it is not reported here, explanation is provided later.
%\subsection{Initial state}
\subsection{Partial result: fidelity regions for each partition}
\label{partial}
According to Lemma \ref{affinite}, all extreme points are just the convex hull of the figure from Figure \ref{fig:plot}. It allows us to use pure states only. In our situation, we use states which are of the form $|\psi^{(2,2)} \> = \begin{bmatrix}
a_1\\
a_2
\end{bmatrix}$ and $|\psi^{(3,1)} \> = \begin{bmatrix}
a_1\\
a_2\\
a_3
\end{bmatrix}$, $a_i \in \mathcal{R}$, so they are real pure states (see Lemma \ref{real}). Each of this state generates a pure state $\rho^{\lambda}$ of the form $\rho^{(2,2)} = \begin{bmatrix}
a_1^2 & a_1a_2\\
a_1a_2&  a_2^2
\end{bmatrix}$, where $a_1^2+a_2^2=1$, for the partition $\lambda = (2,2)$, and $\rho^{(3,1)} =\begin{bmatrix}
a_1^2 & a_1a_2 & a_1a_3\\
a_1a_2 & a_2^2 & a_2a_3\\
a_1a_3&  a_2a_3 & a_3^2
\end{bmatrix}$, where $a_1^2+a_2^2+a_3^2=1$, for $\lambda = (3,1)$ respectively. Inserting $\rho^{(2,2)}$ and $\rho^{(3,1)}$ into Equation \eqref{FFF} from Lemma \ref{FF}, we obtain the plot\footnote[2]{All plots are obtained using $Mathematica$ software.} in Figure \ref{fig:plot}, with values $F_{12}$, $F_{13}$ and $F_{14}$ on axes.
\begin{figure}[p]
\begin{center}$
\begin{array}{c}
\includegraphics[width=0.45\textwidth, height=0.40\textheight]{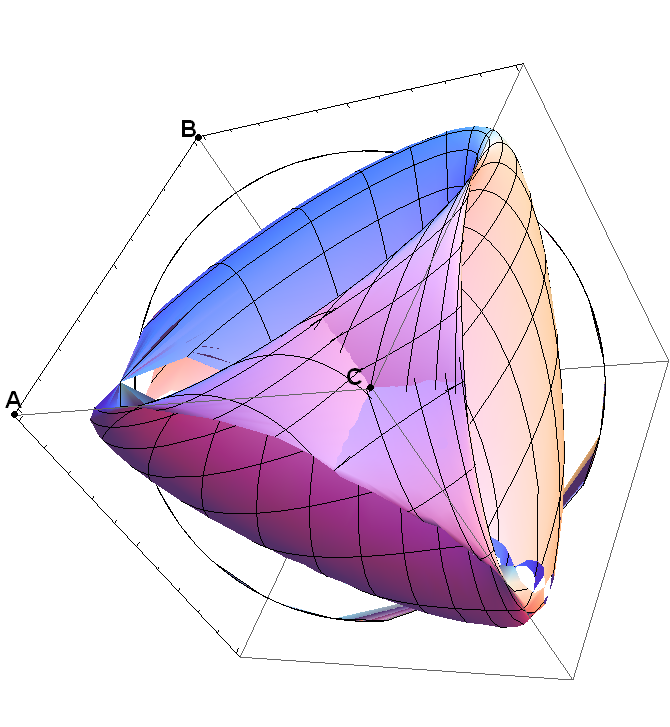} \\
\includegraphics[width=0.45\textwidth, height=0.40\textheight]{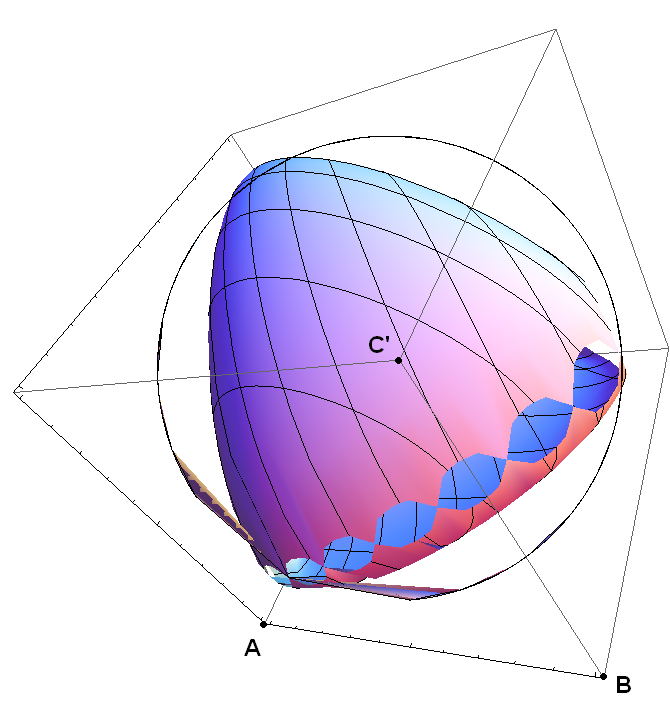}
\end{array}$
\end{center}
\caption{Plot of $F_{1i} = \sum_{\lambda}F_{1i}^{\lambda}$ with $\rho^{(2,2)}$ and $\rho^{(3,1)}$. On axes values of $F_{12}$, $F_{13}$ and $F_{14}$ are reported. The ellipse corresponds to the partition $\lambda = (2,2)$, when the other figure to the partition $(3,1)$. On the ellipse, a small numerical error can be seen. On each plot three coordinates are presented from the set: $A = (0,1,0)$, $B = (1, 1, 0)$ and $C = (0, 0, 0)$, and $C^{'} = (1, 1, 1)$. Two different views of plot are presented.}
\label{fig:plot}
\end{figure}
Now we want to give an explanation, why the partition $\lambda = (4)$ is not included in figures.
It can be shown that in this case, all fidelities are equal to $0$, and since we do not want to 'blur' the main message of this Letter, this partition is excluded from Figure \ref{fig:plot}. But, strictly speaking, in the next step, the convex hull also with this point should be created. This should lead to a 'larger' possible region for fidelities. Of course, it is not a hard task from the computational point of view, since taking a convex hull of some point and some figure is quite an elementary operation.
\begin{remark}
Note that this does not mean that cloning fidelities are equal to $0$ and then that the partition $\lambda = (4)$ should be excluded because perfect universal anti-cloning is forbidden by the laws of quantum mechanics. In our case singlet fractions are equal to $0$ and fidelities are equal $\frac{1}{3}$ (see Eq. \eqref{mipa} and the comment in the paragraph below it).
\end{remark}
At the end of this paragraph we present also explicit formulas for fidelities for every
irrep labeled by partitions $\lambda_2$ and $\lambda_3$.
\begin{itemize}
\item Fidelities for the partition $\lambda_2=(2,2)$:
\be
\begin{split}
F_{12}^{\lambda_3}&=\frac{1}{2}\left(1-a_1^2+a_2^2\right),\\
F_{13}^{\lambda_3}&=\frac{1}{2}\left(1+\frac{a_1^2}{2}-\frac{a_2^2}{2}+\sqrt{3}a_1a_2\right),\\
F_{14}^{\lambda_3}&=\frac{1}{2}\left(1+\frac{a_1^2}{2}-\frac{a_2^2}{2}-\sqrt{3}a_1a_2 \right).
\end{split}
\ee
\item Fidelities for the partition $\lambda_3=(3,1)$:
\be \label{fS4}
\begin{split}
F_{12}^{\lambda_2}&=\frac{1}{2}\left(1-a_1^2+a_2^2+a_3^2\right),\\
F_{13}^{\lambda_2}&=\frac{1}{2}\left(1+\frac{a_1^2}{2}-\frac{a_2^2}{2}+a_3^2+\sqrt{3}a_1a_2\right),\\
F_{14}^{\lambda_2}&=\frac{1}{2}\left(1+\frac{a_1^2}{2}+\frac{5a_2^2}{6}-\frac{a_3^2}{3}+\frac{a_1a_2}{\sqrt{3}}+\frac{2\sqrt{2}a_2a_3}{3}-2\sqrt{\frac{2}{3}}a_1a_3\right).
\end{split}
\ee
\end{itemize}

\subsection{The main result: an allowed region for fidelities}
Since, we have two partitions $(2,2)$ and $(3,1)$ and corresponding pure states $\rho_{13}^{(2,2)}$ and $\rho_{13}^{(3,1)}$, we obtain two figures in Figure \ref{fig:plot}. But, we are interested in situation, where we can obtain a general answer to our question from Section \ref{sec:problem}, namely in situation where we have a mixture of both partitions: $\sum_{\lambda} p_{\lambda} F_{13}^{\lambda}$. To solve this task, we can construct a convex hull of figures from Figure \ref{fig:plot}. It is presented in Figure \ref{fig:ch}.
\begin{figure}[p]
\begin{center}$
\begin{array}{cc}
\includegraphics[width=0.45\textwidth, height=0.40\textheight]{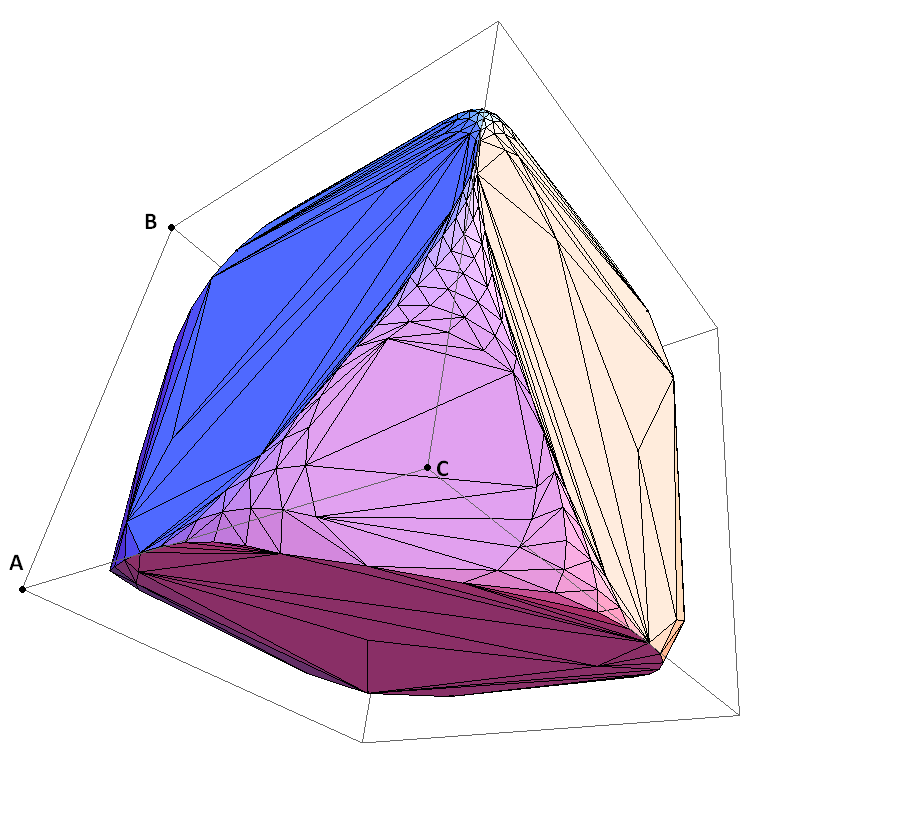} \\
\includegraphics[width=0.45\textwidth, height=0.40\textheight]{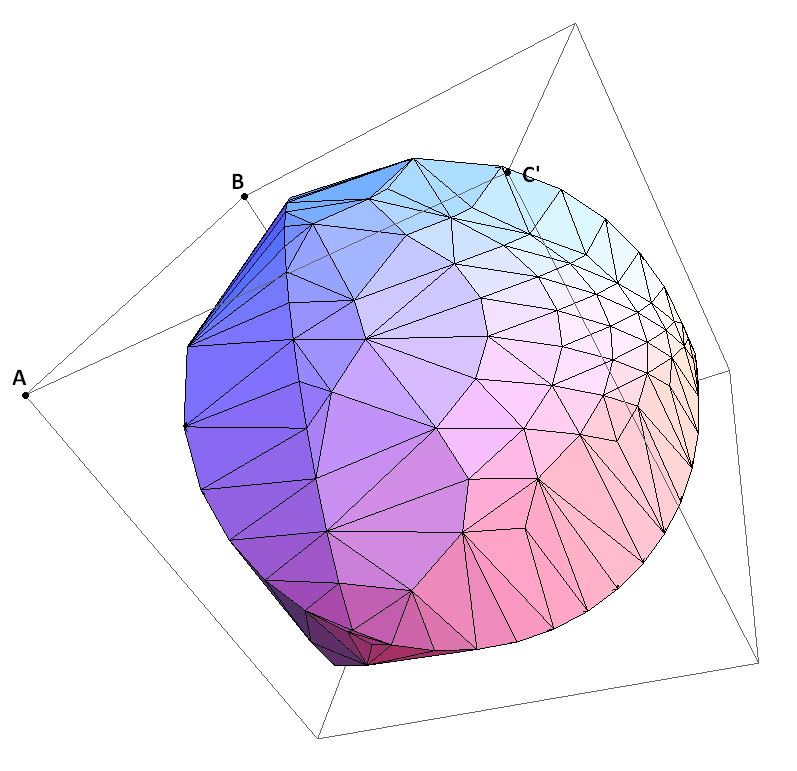}
\end{array}$
\end{center}
\caption{The convex hull of figures that corresponds to the partition $(2,2)$ and $(3,1)$, based on Figure \ref{fig:plot}. On axes values of $F_{12}$, $F_{13}$ and $F_{14}$ are reported. On each plot, three coordinates are presented from the set: $A = (0,1,0)$, $B = (1, 1, 0)$ and $C = (0, 0, 0)$, and $C^{'} = (1, 1, 1)$. Two different views are presented.}
\label{fig:ch}
\end{figure}
One can also see that our Figures \ref{fig:plot} and \ref{fig:ch} are invariant under rotation of angle $2\pi/3$ along straight line $F_{12}=F_{13}=F_{14}$. This corresponds to three conjugacy classes\footnote[3]{Two permutations $\pi$ and $\pi^{'}$ are conjugate iff
$\pi = \sigma \cdot \pi^{'} \cdot \sigma^{-1}$,
where $\sigma$ is also a permutation and $\sigma$, $\pi$, $\pi^{'}$ $\in$ $S_n$.}.
\begin{remark}
Because of the properties of the cloning map $\widetilde{\Lambda}$ (see Sec. \ref{sec:problem}) all possible convex mixtures of the partitions shown in Fig. \ref{fig:ch} are possible and produce a correct quantum cloner, i.e. a trace preserving completely positive map.
\end{remark}

\subsection{Comparison with other models}
As we mentioned before, this particular $UQCM$ - an asymmetric $1 \rightarrow 3$ cloner - was studied before in \cite{Fiurasek-cloning2005} and in \cite{Iblisdir-cloning2004}. It was shown there how to obtain an optimal region for this particular cloner. Plotting Eq. ($38$) for $d=2$ together with a constraint from Eq. ($32$) from \cite{Fiurasek-cloning2005} and then using Eq. \eqref{mipa} (from this work) to obtain fidelities, one can check that then the optimal cloner corresponds to a region for fidelities obtained for the partition $\lambda=(3,1)$ (but first some unknown at this moment constraint should be applied, leading to a cut-off of the plot of the partition $\lambda=(3,1)$, so it will correspond to the result from \cite{Fiurasek-cloning2005}). It is an important remark that one should keep in mind - in our work optimal cloners are not considered, instead, the whole possible region for a given $N-$tuple of fidelities is obtained that later somehow could be optimized by 'throwing away' some partitions and establishing some constrains. What is more, authors of the paper \cite{Fiurasek-cloning2005} also proved the optimality of the asymmetric $1 \rightarrow 2$ quantum cloner using their technique. This result is also in accordance with our results, one can check that using our method, by taking valid partitions from the symmetric group $S(3)$, the same result could be obtained. In general, expressions for possible fidelities in both cases (so our and from \cite{Fiurasek-cloning2005}) look quite similar (compare for example 'our' Eq. \eqref{fS4} with Eq. ($38$) from \cite{Fiurasek-cloning2005}), but because of the constraint from Eq. ($32$), the optimal fidelities could be obtained, not the full region like in our work, so it would be worthy to explore the direct connection between these two approaches in the future.

\subsection{Special case: symmetric cloning}
It can also be shown that in the case of symmetric cloning, the maximal possible value of the triple $(F_{12}, F_{13}, F_{14})$ is $\left( \frac{2}{3}, \frac{2}{3}, \frac{2}{3} \right)$, so it is in accordance with the formula obtained by Keyl et al. \cite{KW} and Werner \cite{Werner-cloning1998} (first using Equation \eqref{mipa}), namely we are able to find such a point $(F_{12}, F_{13}, F_{14})$ that corresponds to the case of the optimal symmetric fidelity $F$.
To find the maximal possible fidelity (for a triple $(F_{12}, F_{13}, F_{14})$) which corresponds to the case of the optimal symmetric cloning, one method is to find a plane equation which includes the ellipse (which corresponds to the partition $(2,2)$). It can be shown that it is given by the formula $F_{12}+F_{13}+F_{14}=\frac{5}{2}$. It can be obtained, for example, by taking coordinates of points from Figure \ref{fig:plot} and then calculating the plane equation according to the following well-known formula:
\be \label{plaszcz}
\left[\begin{array}{ccc}
x-x_1&y-y_1&z-z_1\\
x_2-x_1&y_2-y_1&z_2-z_1\\
x_3-x_2&y_3-y_1&z_3-z_1
\end{array}\right] = 0.
\ee
So, of course we have to identify $x$, $y$, $z$ with $F_{12}$, $F_{13}$ and $F_{14}$ first, and then insert to \eqref{plaszcz} values of three points $P_{1}(F_{12}, F_{13}, F_{14})$, $P_{2}(F_{12}, F_{13}, F_{14})$ and $P_{3}(F_{12}, F_{13}, F_{14})$ from the ellipse and solve equation which is then obtained.

\subsection{Applications}
\label{sec:application}
The above-mentioned approach allows us to reconstruct states from every subspace $\mathcal{H}_{\lambda}^{\mathcal{S}}$ which satisfies not only the symmetric case $(F_1=F_2=F_3)$ (like in a typical $QCM$) but also more general relations between fidelities. In this section we reconstruct states belonging to the subspaces corresponding to partitions $\lambda_1=(2,2), \lambda_2=(3,1)$ and satisfying the following formula
\begin{equation}
\label{max}
\mathop{\max}\limits_{F_1}\left(F_1+F_3=2F_2 \right).
\end{equation}
%\subsection{States reconstruction for partition $\lambda_1=(2,2)$}
%Let us consider following state on $\mathcal{H}^{\mathcal{S}}_{\lambda_1}$
%\begin{equation}
%\rho^{\lambda_1}=\begin{pmatrix}a^2 & ab \\ ab & b^2\end{pmatrix}
%\label{state1}
%\end{equation}
{\bf a) Partition $\lambda_1=(2,2)$}\\
In this case we have two constrains $F_1+F_3=2F_2$ and $a_1^2+a_2^2=1$, so after solving this system of equations we obtain four allowed pairs of solutions
\begin{equation}
\begin{split}
(a_1,a_2)_{(1)}&=\left(\sqrt{\frac{1}{2}+\frac{\sqrt{3}}{4}},-\frac{1}{2}\sqrt{2-\sqrt{3}}\right), \\
(a_1,a_2)_{(2)}&=\left(-\frac{1}{2}\sqrt{2-\sqrt{3}},\sqrt{\frac{1}{2}+\frac{\sqrt{3}}{4}}\right),\\
(a_1,a_2)_{(3)}&=\left(-\sqrt{\frac{1}{2}+\frac{\sqrt{3}}{4}}, -\frac{1}{2}\sqrt{2-\sqrt{3}}\right),\\
(a_1,a_2)_{(4)}&=\left(\frac{1}{2}\sqrt{2-\sqrt{3}},\sqrt{\frac{1}{2}+\frac{\sqrt{3}}{4}} \right).
\end{split}
\end{equation}
Fidelities corresponding to the above pairs are
\begin{equation}
\begin{split}
(F_1^{\lambda_1},F_2^{\lambda_1},F_3^{\lambda_1})_{(1)}&=\left(\frac{1}{4}(2-\sqrt{3}),\frac{1}{2},\frac{1}{4}(2+\sqrt{3})\right),\\ (F_1^{\lambda_1},F_2^{\lambda_1},F_3^{\lambda_1})_{(2)}&=\left(\frac{1}{4}(2+\sqrt{3}),\frac{1}{2},\frac{1}{4}(2-\sqrt{3})\right),\\
(F_1^{\lambda_1},F_2^{\lambda_1},F_3^{\lambda_1})_{(3)}&=\left(\frac{1}{4}(2-\sqrt{3}),\frac{1}{2},\frac{1}{4}(2+\sqrt{3})\right),\\ (F_1^{\lambda_1},F_2^{\lambda_1},F_3^{\lambda_1})_{(4)}&=\left(\frac{1}{4}(2+\sqrt{3}),\frac{1}{2},\frac{1}{4}(2-\sqrt{3})\right).
\end{split}
\end{equation}
We can see that the maximal fidelity $F_1^{\lambda_1}$ can be obtained for a pair $(2)$ and $(4)$, so our reconstructed states are the following
\begin{equation}
\begin{split}
\rho^{\lambda_1}_{(2)}&=\frac{1}{4}\begin{pmatrix}2-\sqrt{3} & -1 \\ -1 & 2+\sqrt{3} \end{pmatrix},\\
 \rho^{\lambda_1}_{(4)}&=\frac{1}{4}\begin{pmatrix}2-\sqrt{3} & 1 \\ 1 & 2+\sqrt{3} \end{pmatrix}.
\end{split}
\end{equation}
Now note something more general. Thanks to Lemma \ref{FF}, we have  $F_1^{\lambda_1}=a_2^2$ and then $a_1^2=1-a_2^2=1-F_1^{\lambda_1}$, so we can express every states  $\rho^{\lambda_1}\in \mathcal{H}_{\lambda_1}^{\mathcal{S}}$ in terms of the fidelity $F_1^{\lambda_1}$
\begin{equation}
\label{signs}
\rho^{\lambda_1}=\begin{pmatrix}1-F_1^{\lambda_1} & \pm \sqrt{F_1^{\lambda_1}(1-F_1^{\lambda_1})}\\ \pm \sqrt{F_1^{\lambda_1}(1-F_1^{\lambda_1} )} & F_1^{\lambda_1} \end{pmatrix}.
\end{equation}
For example, the sign $'+'$ in~\eqref{signs} corresponds to states obtained from pairs $(a_4,b_4)$ and $(a_3,b_3)$ while the sign $'-'$ corresponds to states obtain from pairs $(a_1,b_1)$, $(a_2,b_2)$.\\
{\bf b) Partition $\lambda_2=(3,1)$}\\
For the partition $\lambda_2=(3,1)$ we have a more complicated situation. Here we have three parameters $a_1,a_2,a_3$ but only two constrains $F_1+F_3=2F_2$ and $a_1^2+a_2^2+a_3^2=1$, so we first need to eliminate the parameter $a_1$ and then express the parameter $a_3$ as a function of $a_2$. Numerically, we find that the fidelity $F_1$ is maximal for the following values of parameters $a_1,a_2,a_3$
\begin{equation}
\label{param2}
\begin{split}
(a_1,a_2,a_3)_{(1)}&=\left(0.114, 0.318, 0.941 \right),\\
(a_1,a_2,a_3)_{(2)}&=\left(-0.114, -0.318, -0.941\right),\\
(a_1,a_2,a_3)_{(3)}&=\left(0.114, -0.318, -0.941\right),\\
(a_1,a_2,a_3)_{(4)}&=\left(-0.114, 0.318, 0.941\right).
\end{split}
\end{equation}
Corresponding states are
\begin{equation}
\begin{split}
\rho^{\lambda_2}_{(1)}&=\rho^{\lambda_2}_{(2)}=\begin{pmatrix} 0.013 & 0.036 & 0.107\\
0.036 & 0.101 & 0.299\\
0.107 &  0.299 & 0.886\end{pmatrix},\\
\rho^{\lambda_2}_{(3)}&=\rho^{\lambda_2}_{(4)}=\begin{pmatrix} 0.013 & -0.036 & -0.107\\
-0.036 & 0.101 & 0.299 \\
-0.107 &  0.299 & 0.886\end{pmatrix}.
\end{split}
\end{equation}
Note that fidelities in this case are $\left(F_1^{\lambda_2},F_2^{\lambda_2},F_3^{\lambda_2}\right)=\left(0.886, 0.556, 0.220 \right)$.
The numbers for our numerical data are obtained by Wolfram Mathematica Alpha \cite{math}. Most likely they are indeed optimal.

\section{Fidelities formulas for a $1 \rightarrow 4$ universal quantum cloner}
\label{sec:S5}

In this section we briefly report formulas for fidelities of a $1 \rightarrow 4$  $UQCM$ for every allowed irrep $\lambda$ for an $S(5)$ symmetric group in the case of qubits. In this particular example we are left with three partitions only i.e. $\lambda_1=(5)$, $\lambda_2=(4,1)$ and $\lambda_3=(3,2)$. For the partition $\lambda_2$ unitary representations of transpositions $\operatorname{T}(1,2), \operatorname{T}(1,3), \operatorname{T}(1,4), \operatorname{T}(1,5)$, which we shall need, are~\cite{Thomas1980}:
\be
\begin{split}
\widetilde{\mbV}_{(12)}^{\lambda_2}&=\begin{bmatrix}
1 & 0 & 0 & 0\\
0 & 1 & 0 & 0\\
0 & 0 & 1 & 0\\
0 & 0  & 0 & -1
\end{bmatrix},
\widetilde{\mbV}_{(13)}^{\lambda_2}=\begin{bmatrix}
1 & 0 & 0 & 0\\
0 & 1 & 0 & 0\\
0 & 0 & -\frac{1}{2} & -\frac{\sqrt{3}}{2}\\
0 & 0 & -\frac{\sqrt{3}}{2} & \frac{1}{2}
\end{bmatrix},\\
\widetilde{\mbV}_{(14)}^{\lambda_2}&=\begin{bmatrix}
1 & 0 & 0 & 0 \\
0 & -\frac{1}{3} & -\frac{\sqrt{2}}{3} & -\sqrt{\frac{2}{3}}\\
0 & -\frac{\sqrt{2}}{3} & \frac{5}{6} & -\frac{1}{2\sqrt{3}}\\
0 & -\sqrt{\frac{2}{3}} & -\frac{1}{2\sqrt{3}} & \frac{1}{2}
\end{bmatrix},\\
\widetilde{\mbV}_{(15)}^{\lambda_2}&=\begin{bmatrix}
-\frac{1}{4} & -\frac{1}{4}\sqrt{\frac{5}{3}} & -\frac{1}{2}\sqrt{\frac{5}{6}} & -\frac{1}{2}\sqrt{\frac{5}{2}}\\
-\frac{1}{4}\sqrt{\frac{5}{3}} & \frac{11}{12} & -\frac{1}{6\sqrt{2}} & -\frac{1}{2\sqrt{6}}\\
-\frac{1}{2}\sqrt{\frac{5}{6}} & -\frac{1}{6\sqrt{2}} & \frac{5}{6} & -\frac{1}{2\sqrt{3}}\\
-\frac{1}{2}\sqrt{\frac{5}{2}} & -\frac{1}{2\sqrt{6}} & -\frac{1}{2\sqrt{3}} & \frac{1}{2}
\end{bmatrix}.
\end{split}
\ee
Now, assuming that our pure state is of the form $|\psi^{(4,1)}\rangle=[a_1,a_2,a_3,a_4]^{\operatorname{T}}$, we obtain from Lemma~\ref{FF} the following formulas for fidelities:
\be
\begin{split}
F_{12}^{\lambda_2}&=\frac{1}{2}\left( 1-a_1^2-a_2^2-a_3^2+a_4^2\right),\\
F_{13}^{\lambda_2}&=\frac{1}{2}\left(1-a_1^2-a_2^2+\frac{a_3^2}{2}+\sqrt{3}a_3a_4-\frac{a_4^2}{2}\right),\\
F_{14}^{\lambda_2}&=\frac{1}{2}\left(1-a_1^2+\frac{a_2^2}{3}+\frac{2\sqrt{2}a_2a_3}{3}-\frac{5a_3^2}{6}+2\sqrt{\frac{2}{3}}a_2a_4+\right.\\
&\left. +\frac{a_3a_4}{\sqrt{3}}-\frac{a_4^2}{2}\right),\\
F_{15}^{\lambda_2}&=\frac{1}{2}\left(1+\frac{a_1^2}{4}+\frac{1}{2}\sqrt{\frac{5}{3}}a_1a_2-\frac{11a_2^2}{12}+\sqrt{\frac{5}{6}}a_1a_3+\right.\\
&\left.+ \frac{a_2a_3}{3\sqrt{2}}-\frac{5a_3^2}{6}+\sqrt{\frac{5}{6}}a_1a_4+\frac{a_2a_4}{\sqrt{6}}+\frac{a_3a_4}{\sqrt{3}}-\frac{a_4^2}{2}  \right),
\end{split}
\ee
with a normalization condition  $\sum_{i=1}^4a_i^2=1$.\\
For the partition $\lambda_3$ unitary representations of transpositions $\operatorname{T}(1,2), \operatorname{T}(1,3), \operatorname{T}(1,4), \operatorname{T}(1,5)$, which we shall need, are~\cite{Thomas1980}:
\be
\begin{split}
\widetilde{\mbV}_{(12)}^{\lambda_3}&=\begin{bmatrix}
1 & 0 & 0 & 0 & 0\\
0 & 1 & 0 & 0 & 0\\
0 & 0 & 1 & 0 & 0\\
0 & 0 & 0 & 1 & 0\\
0 & 0 & 0 & 0 & -1
\end{bmatrix},\\
\widetilde{\mbV}_{(13)}^{\lambda_3}&=\begin{bmatrix}
1 & 0 & 0 & 0 & 0\\
0 & -\frac{1}{2} & -\frac{\sqrt{3}}{2} & 0 & 0\\
0 & -\frac{\sqrt{3}}{2} & \frac{1}{2} & 0 & 0\\
0 & 0 & 0 & -\frac{1}{2} & \frac{\sqrt{3}}{2}\\
0 & 0 & 0 & -\frac{\sqrt{3}}{2} & \frac{1}{2}
\end{bmatrix},\\
\widetilde{\mbV}_{(14)}^{\lambda_3}&=\begin{bmatrix}
-\frac{1}{3} & -\frac{\sqrt{2}}{3} & -\sqrt{\frac{2}{3}} & 0 & 0\\
-\frac{\sqrt{2}}{3} & \frac{5}{6} & -\frac{1}{2\sqrt{3}} & 0 & 0\\
-\sqrt{\frac{2}{3}} & -\frac{1}{2\sqrt{3}} & \frac{1}{2} & 0 & 0\\
0 & 0 & 0 & -\frac{1}{2} & \frac{\sqrt{3}}{2}\\
0 & 0 & 0 & \frac{\sqrt{3}}{2} & \frac{1}{2}
\end{bmatrix},\\
\widetilde{\mbV}_{(15)}^{\lambda_3}&=\begin{bmatrix}
-\frac{1}{3} & \frac{1}{3\sqrt{2}} & \frac{1}{\sqrt{6}} & -\frac{1}{\sqrt{6}} & -\frac{1}{\sqrt{2}}\\
\frac{1}{3\sqrt{2}} & -\frac{1}{6} & \frac{1}{\sqrt{3}} & -\frac{1}{\sqrt{3}} & \frac{1}{2}\\
\frac{1}{\sqrt{6}} & \frac{1}{\sqrt{3}} & \frac{1}{2} & \frac{1}{2} & 0\\
-\frac{1}{\sqrt{6}} & -\frac{1}{\sqrt{3}} & \frac{1}{2} & \frac{1}{2} & 0\\
-\frac{1}{\sqrt{2}} & \frac{1}{2} & 0 & 0 & \frac{1}{2}
\end{bmatrix}.
\end{split}
\ee
Assuming that our pure state in this partition has a form $|\psi^{(3,2)}\rangle=[a_1,a_2,a_3,a_4,a_5]^{\operatorname{T}}$ we obtain from Lemma~\ref{FF} the following formulas for fidelities:
\be
\begin{split}
F_{12}^{\lambda_3}&=\frac{1}{2}\left(1-a_1^2-a_2^2-a_3^2-a_4^2+a_5^2\right),\\
F_{13}^{\lambda_3}&=\frac{1}{2}\left(1-a_1^2+\frac{a_2^2}{2}+\sqrt{3}a_2a_3-\frac{a_3^2}{2}+\frac{a_4^2}{2}+\right. \\
&\left. +\sqrt{3}a_4a_5-\frac{a_5^2}{2}\right),\\
F_{14}^{\lambda_3}&=\frac{1}{2}\left(1+\frac{a_1^2}{3}+\frac{2\sqrt{2}a_1a_2}{3}-\frac{5a_2^2}{6}+2\sqrt{\frac{2}{3}}a_1a_3+\right.\\
&\left.+\frac{a_2a_3}{\sqrt{3}}-\frac{a_3^2}{2}+\frac{a_4^2}{2}-\sqrt{3}a_4a_5-\frac{a_5^2}{2}\right),\\
F_{15}^{\lambda_3}&=\frac{1}{2}\left(1+\frac{a_1^2}{3}-\frac{\sqrt{2}a_1a_2}{3}+\frac{a_2^2}{6}-\sqrt{\frac{2}{3}}a_1a_3-\frac{2a_2a_3}{\sqrt{3}}-\right.\\
&\left.+\frac{a_3^2}{2}+\sqrt{\frac{2}{3}}a_1a_4+\frac{2a_2a_4}{\sqrt{3}}-a_3a_4-\frac{a_4^2}{2}+\sqrt{2}a_1a_5-\right.\\
&\left.+a_2a_5-\frac{a_5^2}{2} \right),
\end{split}
\ee
with a normalization condition $\sum_{i=1}^5a_i^2=1$. After taking the convex hull, the allowed region for fidelities could be obtained also for this quantum cloner. Of course, the state reconstruction technique from Section \ref{sec:application} could also be used in this case. At the end of this section, let us note that here we also do not consider partition $\lambda_1$ (see Section~\ref{partial}).

\section{Conclusions}
We have shown that by using representation theory, especially Young diagrams, action of the universal $1 \rightarrow N$ quantum cloning machine could be described. The method of irreps is quite powerful, because it allows us to decompose (usually big) Hilbert space into blocks (linked with a given partition $\lambda$) of smaller dimensions which, of course, are easier to deal with. For example, in our case-study example $1 \rightarrow 3$ $UQCM$, the Hilbert space $\mathbb{C}^{16}$ is decomposed into blocks of dimension $1$, $2$ and $3$ respectively. What is more, using our model, fidelity expressions are quite easy to obtain, one only needs to know representations of all possible irreps for a given symmetric group.

We have also shown that the convex hull could be made to obtain full knowledge about our model, in the sense that it gives raise to the full possible range of fidelities. After careful studies, also optimal $UQCM$ could be found.  What is more, we have shown that by restricting ones attention only to real pure states in each of the block, the full answer to our initial question is obtained. We point out that our $UQCM$ gives correct results for the case of the symmetric cloning and we have also proved that the optimal value, in the symmetric case, could be obtained. What is more important, our approach allows to reconstruct any given state connected to a given $N-$tuple of fidelities.

Of course, in future, it would be interesting to extend our method in such a way that also qudits could be described. What is more, it would be interesting to search for method that allows to obtain the optimal $N$-tuple of fidelities - comparison of our results with these obtained in \cite{Fiurasek-cloning2005} may suggest that the optimal fidelities always correspond to only one partition (with some cut-offs) $\lambda$, since for cases $1 \rightarrow 2$ and $1 \rightarrow 3$ we have that in the first case, the optimal fidelity region corresponds to $\lambda = (2,1)$, in the second one to $\lambda = (3,1)$. Also it is worthy to explore the direct relation between these techniques, so namely 'our' technique and this reported in \cite{Fiurasek-cloning2005}.

{\it Note added}: After the completion of this paper we became aware that similar results for a $1 \rightarrow 3$ universal quantum cloner have been reported in \cite{Yu2010-cloning}.

{\bf Acknowledgments:} We would like to thank Pawe{\l} Horodecki for many valuables discussions and comments on this Letter. M. S. also would like to thank Piotr Migda{\l} for discussions. Last but not least, the anonymous referee's comments are acknowledged. P. \'C., M. H. and M. S. are supported by Polish Ministry of Science and Higher Education grant No. N202 231937. M. S. is also supported by the International PhD Project "Physics of future quantum-based information technologies": grant MPD/2009-3/4 from Foundation for Polish Science. M.H. is also supported by EC IP Q-ESSENCE and ERC grant QOLAPS. Part of this work was done in National Quantum Information Centre of Gda\'nsk.
\section{Appendix}
To prove our results we need the following definitions and lemmas:
\begin{definition}
A set $Z$ is convex iff \cite{hand}:
$\forall_{X,Y \ \in \ Z} \ \ \forall_{0 \leq \mu \leq 1} \ \ \mu X + (1 - \mu) Y \in Z$.
\end{definition}
\begin{definition}
$L(\cdot): x \rightarrow y$ is affine map iff \label{xxx}
$\exists_a \forall_x L(x) = \widetilde{L}(x) + a$,
where $\widetilde{L}$ is linear.
\end{definition}
Now let us present two useful lemmas
\begin{lemma} \label{lemma:af} Suppose that $y$, $y^{'}$ $\in$ $L(\Omega)$, then:
	\be \alpha y + (1 - \alpha) y^{'} \in L(\Omega) \label{lew} \ee (it means that $L(\Omega)$ is a convex set). \end{lemma}
	\begin{proof}
	Let us write:
	\be
       \label{zew}
\begin{split}
	L(x) &= y = L(\sum_i p_i x_i) = \sum_i p_i\widetilde{L}(x_i) + a,  \\
	L(x^{'}) &= y^{'} = L(\sum_i q_i x^{'}_i) = \sum_i q_i\widetilde{L}(x^{'}_i) + a,
\end{split}
\ee
where we use the fact that $L$ is affine and $x = \sum_i p_i x_i$, $x^{'} = \sum_i q_i x^{'}$ for some $x_i, x^{'}_i \in E(\Omega)$. Inserting \eqref{zew} into \eqref{lew} we get:
\be
\begin{split}
&\alpha (\sum_i p_i\widetilde{L}(x_i) + a) + (1- \alpha)(\sum_i q_i\widetilde{L}(x^{'}_i) + a) \\
&= \alpha a + (1 - \alpha)a + \alpha \sum_i \widetilde{L}(p_i x_i) + (1 - \alpha)\sum_i \widetilde{L}(q_i x^{'}_i)  \\
&= a + \widetilde{L}(\alpha x + (1 - \alpha) x^{'}) = L(\alpha x + (1 - \alpha) x^{'}) \in L(\Omega).
\end{split}
\ee
This ends the proof of lemma.
\end{proof}
\begin{lemma}
\label{affinite}
Suppose that we have an affine map $L: \ X \rightarrow Y$. If by $\Omega$ we denote a convex subset of $X$, where $X$ is a finite dimensional space and by $E(\Omega)$ a set of extreme points of $\Omega$, then $L(\Omega)$ is a convex set and $L(E(\Omega))$ can reproduce set $L(\Omega)$ after taking the convex hull, i.e.:
\be L(\Omega) = \operatorname{conv} \ L(E(\Omega)). \ee
\end{lemma}
\begin{proof}
We can write $L(E(\Omega)) \subseteq L(\Omega)$, because we know that $E(\Omega) \subseteq \Omega$. This together with Lemma (\ref{lemma:af}) ends our proof.
\end{proof}

\begin{fact}
\label{Had}
The Hadamard product of two positive-definite matrices is again positive-definite \cite{Horn1985-book-mat}.
\end{fact}
\begin{fact}
\label{RS}
If $A$ is a real symmetric matrix, then all of its eigenvalues are real, and eigenvectors can always be chosen to be real.
\end{fact}
\begin{proof}
Here, we only prove the second part of Fact \ref{RS} since the first one is obvious. Let $\lambda$ be a real eigenvalue of the matrix $A$ with an associated eigenvector $v$. Let us write $v$ as $v = \Re(v) + i\Im(v)$, so that $A(\Re(v) + \Im(v)) = \lambda(\Re(v) + \Im(v))$. It implies that $A\Re(v) = \lambda\Re(v)$ and $A\Im(v) = \lambda\Im(v)$. Now, if $\Re(v) \neq 0$, then it is a real eigenvector of $A$. When $\Re(v)=0$, then $\Im(v)$ is a real eigenvector.
\end{proof}

\newpage
\addcontentsline{toc}{section}{\bf Bibliography}
\bibliographystyle{apsrev}
\bibliography{mag_PCbib}
%\bibliography{rmp13-hugekey}
% czyta z katalogu Miktex/bibtex/rmp !!!!!
\end{document}